\documentclass[a4paper]{article}
\pdfoutput=1

\usepackage[utf8]{inputenc}
\usepackage[T1]{fontenc}
\usepackage[american]{babel}

\usepackage[final]{microtype}

\usepackage{amsfonts}
\usepackage{amsmath}
\usepackage{amssymb}
\usepackage[backend=biber,style=authoryear-comp, maxcitenames=9999,
    sorting=ynt]{biblatex}
\usepackage{color}
\usepackage{complexity} 
\usepackage{csquotes}
\usepackage{float} 
\usepackage{hyperref} 
\usepackage{hyphenat} 
\usepackage{mathtools} 
\usepackage{tikz}

\usepackage{amthm}

\usepackage{geometry}
\usepackage[final]{microtype}

\AtEveryBibitem{%
  \clearfield{doi}
  \clearfield{isbn}
  \clearfield{issn}
  \clearfield{url}
}

\newcommand{\Exp}{\text{X}} 
\newcommand{\Nz}{\mathbb{N}_0} 
\newcommand{\N}{\mathbb{N}_+} 
\newcommand{\Qacc}{Q_{\textrm{acc}}}
\newcommand{\Qrej}{Q_{\textrm{rej}}}
\newcommand{\Z}{\mathbb{Z}} 
\newcommand{\domega}{\omega\omega} 

\newlength{\inaltlen}
\newcommand{\inalt}{%
  \mathpunct{}\mathord{\hbox{\vrule height \inaltlen \vbox to \inaltlen%
  {\hrule width \inaltlen height 0.4pt \vss%
  \hrule width \inaltlen height 0.4pt \vss%
  \hrule width \inaltlen height 0.4pt}}}%
}

\mathchardef\mhyphen="2D

\newclass{\Acl}{A}
\newclass{\Bcl}{B}
\newclass{\CAP}{CAP} 
\newclass{\CATIME}{CATIME} 
\newclass{\EXCAP}{EXCAP} 
\newclass{\EXPTIME}{EXPTIME} 
\newclass{\MARXCAP}{MAR\mhyphen XCAP} 
\newclass{\PTIME}{P} 
\newclass{\SCAP}{SCAP} 
\newclass{\SCATIME}{SCATIME} 
\newclass{\SXCAP}{SXCAP} 
\newclass{\SXCATIME}{SXCATIME} 
\newclass{\SimulNP}{SimulNP} 
\newclass{\VXCAP}{VXCAP} 
\newclass{\XCAP}{XCAP} 
\newclass{\XCATIME}{XCATIME} 
\newclass{\oCAP}{1CAP} 
\newclass{\oCATIME}{1CATIME} 
\newclass{\oXCAP}{1XCAP} 
\newclass{\oXCATIME}{1XCATIME} 

\newlang{\BOOL}{BOOL} 
\newlang{\STATE}{STATE} 
\newlang{\TAUT}{TAUT} 
\newlang{\TQBF}{TQBF} 
\newlang{\UNSAT}{UNSAT} 

\newcommand{\MARSTATEALL}{\STATE_\forall^\text{MAR}}
\newcommand{\SATTAUT}{\SAT^\land\mhyphen\TAUT^\lor}
\newcommand{\STATEALL}{\STATE_\forall}
\newcommand{\STATEONE}{\STATE_\exists}
\newcommand{\TAUTSAT}{\TAUT^\land\mhyphen\SAT^\lor}
\newcommand{\ttpNP}{{\le_{tt}^p}(\NP)} 

\usetikzlibrary{arrows,automata,calc,positioning}

\begin{document}


\title{Complexity-Theoretic Aspects of Expanding Cellular Automata%
  \thanks{Parts of this paper have been submitted \autocite{BA, MA} in partial
fulfillment of the requirements for the degrees of Bachelor of Science and
Master of Science at the Karlsruhe Institute of Technology (KIT).}
}
\author{Augusto Modanese \\
  Karlsruhe Institute of Technology (KIT), Germany\\
  \texttt{modanese@kit.edu}
}
\date{}

\maketitle



\begin{abstract}
  The expanding cellular automata (XCA) variant of cellular automata is
investigated and characterized from a complexity-theoretical standpoint.
  An XCA is a one-dimensional cellular automaton which can dynamically create
new cells between existing ones.
  The respective polynomial-time complexity class is shown to coincide with
$\ttpNP$, that is, the class of decision problems polynomial-time truth-table
reducible to problems in $\NP$.
  An alternative characterization based on a variant of non-deterministic Turing
machines is also given.
  In addition, corollaries on select XCA variants are proven:
  XCAs with multiple accept and reject states are shown to be polynomial-time
equivalent to the original XCA model.
  Finally, XCAs with alternative acceptance conditions are considered and
classified in terms of $\ttpNP$ and the Turing machine polynomial-time class
$\PTIME$.
\end{abstract}

\section{Introduction}

Traditionally, cellular automata (CAs) are defined as a rigid and immutable
lattice of cells; their behavior is dictated exclusively by a local transition
function operating on homogeneous local configurations.
This can be generalized, for instance, by mutable neighborhoods
\autocite{rosenfeld81} or by endowing CAs with the ability to \emph{shrink},
that is, delete their cells \autocite{rosenfeld83}.
When shrinking, the automaton's structure and dimension are preserved by
\enquote{gluing} the severed parts and reconnecting their delimiting cells as
neighbors.
When employed as language recognizers, shrinking CAs (SCAs) can be more
efficient than standard CAs \autocite{rosenfeld83, kutrib15}.

Other variants of CAs with dynamically reconfigurable neighborhoods have
emerged throughout the years.
In the case of two-dimensional CAs, there is the structurally dynamical CA
(SDCA) due to \textcite{SDCA}, in which the connections between neighbors are
created and dropped depending on the local configuration.
In the one-dimensional case, further variants in this sense are considered
in the work of \textcite{dubacq_dyn_neighborhood}, where one finds, in
particular, CAs whose neighborhoods vary over time.
Dubacq also proposes the dynamically reconfigurable CA (DRCA), a CA whose cells
are able to exchange their neighbors for neighbors of their neighbors.
\textcite{dantchev_dyn_neighborhood} later points out a drawback in the
definition of DRCAs and proposes an alternative dubbed the dynamic neighborhood
CA (DNCA).

By relaxing the arrangement of cells as a lattice, CAs may be generalized to
graphs \autocite{graph_automata}.
Graph automata are related to CAs in that each vertex in the graph corresponds
to a cell; thus, graphs whose vertices have finite degrees provide a natural
generalization of CAs.
\textcite{graph_automata} also define a rule based on topological refinements of
graphs, which may be used as a model for biological cell division.
An additional example of cell division in this sense is the \enquote{inflating
grid} of \textcite{causal_graph_dyn}.

Modeling cell division and growth, in fact, was one of the driving motivations
towards the investigation of the \emph{expanding CA} (XCA) in \textcite{BA}.
An XCA is, in a way, the opposite of an SCA; instead of cells vanishing, new
cells can emerge between existing ones.
This operation is topologically similar to the cell division of graph automata;
as in the SCA model, however, it maintains the overall arrangement and
connectivity of the automaton's cells as similar as possible to that of standard
CAs (i.e., a bi-infinite, one-dimensional array of cells).

We mention a few aspects in which XCAs differ from the aforementioned variants.
Contrary to SDCAs or CAs with dynamic neighborhoods such as DRCAs and DNCAs,
XCAs enable the creation of new cells, not simply new links between existing
ones.
In addition, the XCA model does not focus as much on the reconfiguration of
cells; in it, the neighborhoods are homogeneous and predominantly immutable.
Furthermore, in contrast to the far more general graph automata, XCAs are still
one-dimensional CAs; this ensures basic CA
techniques (e.g., synchronization) function the same as they do in standard CAs.

Finally, shrinking and expanding are not mutually exclusive.
Combining them yields the shrinking and expanding CA (SXCA).
The polynomial-time class of SXCA language deciders was shown
to coincide with $\PSPACE$ \autocite{BA, SXCA}.

A previous result by \textcite{BA} is that, for the class $\XCAP$ of
polynomial-time XCA language deciders, we have $\NP \cup \coNP \subseteq
\XCAP \subseteq \PSPACE$.
A precise characterization of $\XCAP$, however, remained outstanding.
Such was the topic of the author's master's thesis \autocite{MA}, the results of
which are summarized in this paper.
The main result (Theorem~\ref{thm_XCAP_ttpNP}) is $\XCAP$ being equal to the
class of decision problems which are polynomial-time truth-table reducible to
$\NP$, denoted $\ttpNP$.

The rest of the paper is organized as follows:
Section~\ref{sec_definitions} covers the fundamental definitions and results
needed for the subsequent discussions.
Following that, Section~\ref{sec_XCAP} recalls the main result of \textcite{BA}
concerning $\XCAP$ and presents two characterizations of $\XCAP$, one based on
$\ttpNP$ (Theorem~\ref{thm_XCAP_ttpNP}) and another (Theorem~\ref{thm_simulNP})
based on a variant of non-deterministic Turing machines (NTMs).
Section \ref{sec_implications} covers some immediate corollaries, in particular
by considering an XCA variant with multiple accept and reject states as well as
two other variants with diverse acceptance conditions.
Finally, Section~\ref{sec_conclusion} concludes.

This is an extended and revised version of a preliminary paper presented at the
AUTOMATA 2019 \autocite{XCA}.
Section~\ref{sec_XCAP} has been expanded to provide a complete proof of
Proposition~\ref{prop_NP_coNP_in_XCAP} (instead of only an outline), which is
now found in Section~\ref{sec_XCAP_BA}, while the material in
Section~\ref{sec_syncNTM} is entirely novel.
Other improvements include a full proof of Lemma~\ref{lem_xca_number_cells},
an updated abstract, broader discussions of concepts and results, and minor
text edits.

\section{Definitions}
\label{sec_definitions}

This section recalls basic concepts and results needed for the proofs and
discussions in the later sections and is broken down in two parts.
The first is concerned with basic topics regarding formal languages, Turing
machines, and Boolean formulas.
The second part covers the definition of expanding CAs.

\subsection{Formal Languages and Turing Machines}

It is assumed the reader is familiar with the fundamentals of cellular automata
and complexity theory (see, e.g., standard references such as
\cite{delorme99_cellular_book} and \cite{arora09_computational_book}).
Unless stated otherwise, all words have length at least one.
For sets $A$ and $B$, $B^A$ denotes the set of functions $A \to B$.
For an alphabet $\Sigma$, $\Sigma^\ast$ is the set of words over $\Sigma$.
A \emph{$\domega$-word} is a biinfinite word $w\in \Sigma^\Z$.
The notion of a complete language is employed strictly in the sense of
polynomial-time many-one reductions by deterministic Turing machines.

\subsubsection{Boolean Formulas}

Let $V$ be a language of \emph{variables} over an alphabet $\Sigma$ which,
without loss of generality, is disjoint from $\{F, T, \lnot, \land, \lor, (,
)\}$.
$\BOOL_V$ denotes the formal language of Boolean formulas over the variables of
$V$.
For better readability, we shall prefer prefix notation when writing out
formulas (e.g., $\land( f, g )$ for formulas $f$ and $g$ instead of the more
common infix notation $f \land g$).

An \emph{interpretation} of $V$ is a map $I\colon V \to \{F, T\}$.
Each such $I$ gives rise to an \emph{evaluation} $E_I\colon \BOOL_V \to \{F,
T\}$ which, given a formula $f \in \BOOL_V$, substitutes each variable
$x \in V$ with the truth value $I(x)$ and reduces the resulting formula using
standard propositional logic.
A formula $f$ is \emph{satisfiable} if there is an interpretation $I$ such that
$E_I(f) = T$, and $f$ is a \emph{tautology} if this holds for every $I$.

In order to define the languages $\SAT$ of satisfiable formulas and $\TAUT$ of
tautologies, a language $V$ of variables must first be agreed on.
In this paper, variables are encoded as binary strings prefixed by a special
symbol $x$, that is, $V = \{x\} \cdot \{0, 1\}^+$.
The language $\SAT$ contains exactly the satisfiable formulas of $\BOOL_V$.
Similarly, $\TAUT$ contains exactly the tautologies of $\BOOL_V$.
The following is a classical result concerning $\SAT$ and $\TAUT$:
\begin{theorem}[\cite{SATNPcomplete}]
  \label{thm_sat_np_complete}
  $\SAT$ is $\NP$-complete, and $\TAUT$ is $\coNP$-complete.
\end{theorem}

\subsubsection{Truth-Table Reductions}

The theory of truth-table reductions was established by \textcite{ttpNP_I,
ttpNP_II}.
Later, \textcite{ttpNP_eq_bfpNP} showed the class of decision problems
polynomial-time truth-table (i.e., Boolean-circuit) reducible to $\NP$, denoted
$\ttpNP$, remains the same even if the reduction is in terms of Boolean formulas
(instead of circuits).
We refer to \textcite{Buss_ttpNP} for a series of alternative characterizations
of $\ttpNP$.
The inclusions $\NP \cup \coNP \subseteq \ttpNP$ and $\ttpNP \subseteq \PSPACE$
are known to hold.

A more formal treatment of the class $\ttpNP$ is not necessary to establish the
results of this paper; it suffices to note $\ttpNP$ has complete languages.
In particular, we are interested in Boolean formulas with $\NP$ and $\coNP$
predicates.
To this end, we employ $\SAT$ and $\TAUT$ to define membership predicates of
the form $f \inalt_L$, where $f$ is a Boolean formula, $L \in \{ \SAT, \TAUT
\}$, and \enquote{$\inalt_L$} is a purely syntactic construct which stands for
the statement \enquote{$f \in L$}.

\begin{definition}[$\SATTAUT$]
  Let $V = \{x\} \cdot \{0,1\}^+$ and $V_L = \BOOL_V \cdot \{\inalt_L\}$ for $L
\in \{\SAT, \TAUT\}$.
  The language $\BOOL^{\land\lor}_{\SAT,\TAUT} \subseteq \BOOL_{V_\SAT \cup
V_\TAUT}$ is defined recursively as follows:
  \begin{enumerate}
    \item $V_\SAT, V_\TAUT \subseteq \BOOL^{\land\lor}_{\SAT,\TAUT}$.
    \item For $v \in V_\SAT$ and $f \in \BOOL^{\land\lor}_{\SAT,\TAUT}$,
      $\land(v, f) \in \BOOL^{\land\lor}_{\SAT,\TAUT}$.
    \item For $v \in V_\TAUT$ and $f \in \BOOL^{\land\lor}_{\SAT,\TAUT}$,
      $\lor(v, f) \in \BOOL^{\land\lor}_{\SAT,\TAUT}$.
  \end{enumerate}
  The language $\SATTAUT \subseteq \BOOL^{\land\lor}_{\SAT,\TAUT}$ contains all
formulas which are true under the interpretation mapping $f \inalt_L$ to the
truth value of the statement \enquote{$f \in L$}.
  \label{def_sattaut}
\end{definition}

For example, given $f_1, f_2, f_3, f_4 \in \BOOL_V$, the following formula $f$
is in $\BOOL^{\land\lor}_{\SAT,\TAUT}$:
\[
  f = \land( f_1 \inalt_\SAT,
      \lor( f_2 \inalt_\TAUT,
      \land(f_3 \inalt_\SAT, f_4 \inalt_\TAUT))).
\]
Then, $f \in \SATTAUT$ if, for instance, $f_1 \in \SAT$ and $f_2 \in \TAUT$
holds.

From the results of \textcite{Buss_ttpNP} it follows:
\begin{theorem}
  $\SATTAUT$ is $\ttpNP$-complete.
  \label{thm_tautsat_ttpNP}
\end{theorem}

\subsection{Cellular Automata}
\label{sec_CA}

Here, we are strictly interested in one-dimensional cellular automata (CAs) with
the standard neighborhood and employed as language deciders.
CA deciders possess a \emph{quiescent state} $q$; cells which are not in this
state are said to be \emph{active} and may not become quiescent.
The input for a CA decider is provided in its initial configuration surrounded
by quiescent cells.
As deciders, CAs are Turing complete, and, more importantly, CAs can simulate
TMs in real-time \autocite{smith71}.
Conversely, it is known a TM can simulate a CA with time complexity $t$ in time
at most $t^2$.
A corollary is that the class of problems decidable in polynomial time by CAs
is exactly $\PTIME$.

\subsubsection{Expanding Cellular Automata}

First considered in \textcite{BA}, the expanding CA (XCA) is similar to the
shrinking CA (SCA) in that it is dynamically reconfigurable; instead of cells
being deleted, however, in an XCA new cells emerge between existing ones.
This does not alter the underlying topology, which remains one-dimensional and
biinfinite.

For modeling purposes, the new cells are seen as \emph{hidden} between the
original (i.e., \emph{visible}) ones, with one hidden cell placed between any
two neighboring visible cells.
These latter cells serve as the hidden cell's left and right neighbors and are
referred to as its \emph{parents}.
In each CA step, a hidden cell observes the states of its parents and either
assumes a non-hidden state, thus becoming visible, or remains hidden.
In the former case, the cell assumes the position between its parents and
becomes an ordinary cell (i.e., visible), and the parents are reconnected so as
to adopt the new cell as a neighbor.
Visible cells may not become hidden, and we refer to hidden cells neither as
active nor as quiescent (i.e., we treat them as a \emph{tertium quid}).

\begin{definition}[XCA]
  Let $N = \{ -1, 0, 1 \}$ be the standard neighborhood.
  An \emph{expanding CA} (\emph{XCA}) is a CA $A$ with state set $Q$ and
local transition function $\delta\colon Q^N \to Q$ and which possesses a
distinguished \emph{hidden state} $\odot \in Q$.
  For any local configuration $\ell\colon N \to Q$, $\delta(\ell) = \odot$ is
allowed only if $\ell(0) = \odot$.

For a global configuration $c\colon \Z \to Q$, let $h_c\colon \Z \to Q^N$ be
such that $h_c(z)(-1) = c(z)$, $h_c(z)(0) = \odot$, and $h_c(z)(1) = c(z + 1)$
for any $z \in \Z$.
  Define $\alpha\colon Q^\Z \to Q^\Z$ as follows, where $\Delta$ is the standard
CA global transition function (as induced by $\delta$):
  \[
    \alpha(c)(z) = \begin{cases}
      \Delta(c)(\frac{z}{2}), & \text{$z$ even} \\
      \delta(h_c(\frac{z-1}{2})), & \text{otherwise.}
    \end{cases}
  \]
  Finally, with $c$ still arbitrary, let $\Phi\colon Q^\Z \to Q^\Z$ be the map%
  \footnote{Strictly speaking, the codomain of $\Phi$ (as here defined) is not
    only $Q^\Z$ but actually larger (since, for $c \in Q^\Z$ arbitrary,
    $\Phi(c)(z)$ may be undefined for certain $z \in \Z$).
    However, since the configurations that arise in our context of XCAs have
    infinitely many occurrences of $q$ in either direction (i.e., $c(i) = q$
    holds for infinitely many $i > 0$ as well as infinitely many $i < 0$), in
    this case $\Phi(c)(z)$ is guaranteed to be defined for every $z \in \Z$,
    that is, $\Phi(c) \in Q^\Z$.
    Hence, to simplify the presentation, we write only \enquote{$Q^\Z$} here.
  }
  that acts as a homomorphism on $c$ deleting any occurrence of $\odot$
  (and contracting the remaining states towards zero), formally:
  \[
    \Phi(c)(z) = \begin{cases}
      c(m_+(z)), & z \ge 0 \\
      c(m_-(-z-1)), & \text{otherwise}
    \end{cases}
  \]
  where $m_+(z)$ is the maximum $i \in \Z$ for which
  $| \{ j \in [0,i) \mid c(j) \neq \odot \} | = z$
  and $m_-(z)$ is the minimum $i \in \Z$ for which
  $| \{ j \in (i,-1] \mid c(j) \neq \odot \} | = z$.
  Then the global transition function of $A$ is $\Delta^\Exp = \Phi \circ
\alpha$.
\label{def_xca}
\end{definition}

Figure \ref{fig_ex_XCA} illustrates an XCA $A$ and its operation for input
$001010$ as an example.
The local transition function $\delta$ of $A$ is as follows:
\[
  \delta(q_{-1}, q_0, q_1) = \begin{cases}
    q_{-1} \oplus q_1, & q_{-1}, q_1 \in \{0,1\} \\
    q_0, & \text{otherwise}
  \end{cases}
\]
where $\oplus$ denotes the bitwise XOR operation, that is, addition modulo $2$.
The initial configuration is marked as $c$.
The hidden cells are those in state $\odot$.
Starting from $c$, $\alpha$ applies $\delta$ to each local configuration, where
$h_c$ specifies the local configurations for the hidden cells; $\alpha$ also
promotes all originally hidden cells to visible ones.
Finally, $\Phi$ then eliminates cells having the state $\odot$, as such cells
are per definition not allowed to be visible (rather, they are present only
implicitly in the global configuration).

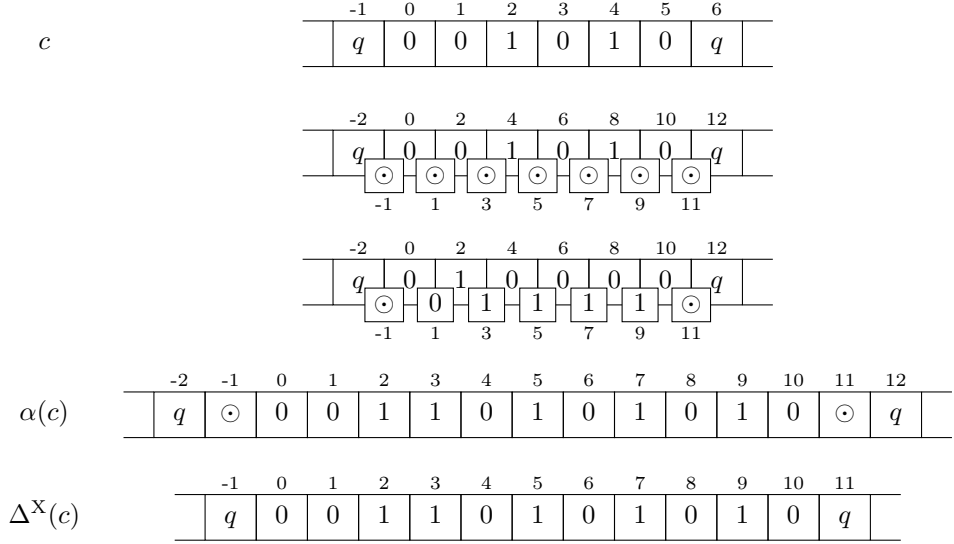
\begin{figure*}
  \centering
  \begin{tikzpicture}[every node/.style={block},
      block/.style={
        minimum width=width("$\odot$") + 4mm,
        minimum height=height("$\odot$") + 4mm,
        text height=1.2ex,
        text depth=.2ex,
        outer sep=0pt,
        draw, rectangle, node distance=0pt}]

    \tikzstyle{h} = [
        minimum width=width("$\odot$") + 2mm,
        minimum height=height("$\odot$") + 2mm,
        fill=white, rectangle
    ]

    \newlength{\rowspc}
    \setlength{\rowspc}{3.15em}

    \node    (z0)                            {$q$};
    \node    (z1)  [right=of z0]             {$0$};
    \node    (z2)  [right=of z1]             {$0$};
    \node    (z3)  [right=of z2]             {$1$};
    \node    (z4)  [right=of z3]             {$0$};
    \node    (z5)  [right=of z4]             {$1$};
    \node    (z6)  [right=of z5]             {$0$};
    \node    (z7)  [right=of z6]             {$q$};
    \draw    (z0.north west) -- ++(-0.4cm,0)
             (z0.south west) -- ++(-0.4cm,0)
             (z7.north east) -- ++(0.4cm,0)
             (z7.south east) -- ++(0.4cm,0);
    \node[draw=none] at (z0.north) [yshift=1.5mm] {\scriptsize -1};
    \node[draw=none] at (z1.north) [yshift=1.5mm] {\scriptsize 0};
    \node[draw=none] at (z2.north) [yshift=1.5mm] {\scriptsize 1};
    \node[draw=none] at (z3.north) [yshift=1.5mm] {\scriptsize 2};
    \node[draw=none] at (z4.north) [yshift=1.5mm] {\scriptsize 3};
    \node[draw=none] at (z5.north) [yshift=1.5mm] {\scriptsize 4};
    \node[draw=none] at (z6.north) [yshift=1.5mm] {\scriptsize 5};
    \node[draw=none] at (z7.north) [yshift=1.5mm] {\scriptsize 6};

    \node    (a0)  [below=of z0, yshift=-2.4em] {$q$};
    \node    (a1)  [right=of a0]                 {$0$};
    \node    (a2)  [right=of a1]                 {$0$};
    \node    (a3)  [right=of a2]                 {$1$};
    \node    (a4)  [right=of a3]                 {$0$};
    \node    (a5)  [right=of a4]                 {$1$};
    \node    (a6)  [right=of a5]                 {$0$};
    \node    (a7)  [right=of a6]                 {$q$};
    \node[h] (a01) [at=     (a0.south east)]     {$\odot$};
    \node[h] (a12) [at=     (a1.south east)]     {$\odot$};
    \node[h] (a23) [at=     (a2.south east)]     {$\odot$};
    \node[h] (a34) [at=     (a3.south east)]     {$\odot$};
    \node[h] (a45) [at=     (a4.south east)]     {$\odot$};
    \node[h] (a56) [at=     (a5.south east)]     {$\odot$};
    \node[h] (a67) [at=     (a6.south east)]     {$\odot$};
    \draw    (a0.north west) -- ++(-0.4cm,0)
             (a0.south west) -- ++(-0.4cm,0)
             (a7.north east) -- ++(0.4cm,0)
             (a7.south east) -- ++(0.4cm,0);
    \node[draw=none] at (a0.north) [yshift=1.5mm] {\scriptsize -2};
    \node[draw=none] at (a1.north) [yshift=1.5mm] {\scriptsize 0};
    \node[draw=none] at (a2.north) [yshift=1.5mm] {\scriptsize 2};
    \node[draw=none] at (a3.north) [yshift=1.5mm] {\scriptsize 4};
    \node[draw=none] at (a4.north) [yshift=1.5mm] {\scriptsize 6};
    \node[draw=none] at (a5.north) [yshift=1.5mm] {\scriptsize 8};
    \node[draw=none] at (a6.north) [yshift=1.5mm] {\scriptsize 10};
    \node[draw=none] at (a7.north) [yshift=1.5mm] {\scriptsize 12};
    \node[draw=none] at (a01.south) [yshift=-1.5mm] {\scriptsize -1};
    \node[draw=none] at (a12.south) [yshift=-1.5mm] {\scriptsize 1};
    \node[draw=none] at (a23.south) [yshift=-1.5mm] {\scriptsize 3};
    \node[draw=none] at (a34.south) [yshift=-1.5mm] {\scriptsize 5};
    \node[draw=none] at (a45.south) [yshift=-1.5mm] {\scriptsize 7};
    \node[draw=none] at (a56.south) [yshift=-1.5mm] {\scriptsize 9};
    \node[draw=none] at (a67.south) [yshift=-1.5mm] {\scriptsize 11};

    \node    (b0)  [below=of a0, yshift=-\rowspc] {$q$};
    \node    (b1)  [right=of b0]                  {$0$};
    \node    (b2)  [right=of b1]                  {$1$};
    \node    (b3)  [right=of b2]                  {$0$};
    \node    (b4)  [right=of b3]                  {$0$};
    \node    (b5)  [right=of b4]                  {$0$};
    \node    (b6)  [right=of b5]                  {$0$};
    \node    (b7)  [right=of b6]                  {$q$};
    \node[h] (b01) [at=     (b0.south east)]      {$\odot$};
    \node[h] (b12) [at=     (b1.south east)]      {$0$};
    \node[h] (b23) [at=     (b2.south east)]      {$1$};
    \node[h] (b34) [at=     (b3.south east)]      {$1$};
    \node[h] (b45) [at=     (b4.south east)]      {$1$};
    \node[h] (b56) [at=     (b5.south east)]      {$1$};
    \node[h] (b67) [at=     (b6.south east)]      {$\odot$};
    \draw    (b0.north west) -- ++(-0.4cm,0)
             (b0.south west) -- ++(-0.4cm,0)
             (b7.north east) -- ++(0.4cm,0)
             (b7.south east) -- ++(0.4cm,0);
    \node[draw=none] at (b0.north) [yshift=1.5mm] {\scriptsize -2};
    \node[draw=none] at (b1.north) [yshift=1.5mm] {\scriptsize 0};
    \node[draw=none] at (b2.north) [yshift=1.5mm] {\scriptsize 2};
    \node[draw=none] at (b3.north) [yshift=1.5mm] {\scriptsize 4};
    \node[draw=none] at (b4.north) [yshift=1.5mm] {\scriptsize 6};
    \node[draw=none] at (b5.north) [yshift=1.5mm] {\scriptsize 8};
    \node[draw=none] at (b6.north) [yshift=1.5mm] {\scriptsize 10};
    \node[draw=none] at (b7.north) [yshift=1.5mm] {\scriptsize 12};
    \node[draw=none] at (b01.south) [yshift=-1.5mm] {\scriptsize -1};
    \node[draw=none] at (b12.south) [yshift=-1.5mm] {\scriptsize 1};
    \node[draw=none] at (b23.south) [yshift=-1.5mm] {\scriptsize 3};
    \node[draw=none] at (b34.south) [yshift=-1.5mm] {\scriptsize 5};
    \node[draw=none] at (b45.south) [yshift=-1.5mm] {\scriptsize 7};
    \node[draw=none] at (b56.south) [yshift=-1.5mm] {\scriptsize 9};
    \node[draw=none] at (b67.south) [yshift=-1.5mm] {\scriptsize 11};

    \tikzstyle{s} = [
        minimum width=width("$\odot$") + 4mm,
        minimum height=height("$\odot$") + 4mm
    ]

    \node[s] (c34) [below=of b34, yshift=-\rowspc+0.5em]  {$1$};
    \node[s] (c4)  [right=of c34]                   {$0$};
    \node[s] (c45) [right=of c4]                    {$1$};
    \node[s] (c5)  [right=of c45]                   {$0$};
    \node[s] (c56) [right=of c5]                    {$1$};
    \node[s] (c6)  [right=of c56]                   {$0$};
    \node[s] (c67) [right=of c6]                    {$\odot$};
    \node[s] (c7)  [right=of c67]                   {$q$};
    \node[s] (c3)  [left=of  c34]                   {$0$};
    \node[s] (c23) [left=of  c3]                    {$1$};
    \node[s] (c2)  [left=of  c23]                   {$1$};
    \node[s] (c12) [left=of  c2]                    {$0$};
    \node[s] (c1)  [left=of  c12]                   {$0$};
    \node[s] (c01) [left=of  c1]                    {$\odot$};
    \node[s] (c0)  [left=of  c01]                   {$q$};
    \draw    (c0.north west) -- ++(-0.4cm,0)
             (c0.south west) -- ++(-0.4cm,0)
             (c7.north east) -- ++(0.4cm,0)
             (c7.south east) -- ++(0.4cm,0);
    \node[draw=none] at (c0.north) [yshift=1.5mm] {\scriptsize -2};
    \node[draw=none] at (c01.north) [yshift=1.5mm] {\scriptsize -1};
    \node[draw=none] at (c1.north) [yshift=1.5mm] {\scriptsize 0};
    \node[draw=none] at (c12.north) [yshift=1.5mm] {\scriptsize 1};
    \node[draw=none] at (c2.north) [yshift=1.5mm] {\scriptsize 2};
    \node[draw=none] at (c23.north) [yshift=1.5mm] {\scriptsize 3};
    \node[draw=none] at (c3.north) [yshift=1.5mm] {\scriptsize 4};
    \node[draw=none] at (c34.north) [yshift=1.5mm] {\scriptsize 5};
    \node[draw=none] at (c4.north) [yshift=1.5mm] {\scriptsize 6};
    \node[draw=none] at (c45.north) [yshift=1.5mm] {\scriptsize 7};
    \node[draw=none] at (c5.north) [yshift=1.5mm] {\scriptsize 8};
    \node[draw=none] at (c56.north) [yshift=1.5mm] {\scriptsize 9};
    \node[draw=none] at (c6.north) [yshift=1.5mm] {\scriptsize 10};
    \node[draw=none] at (c67.north) [yshift=1.5mm] {\scriptsize 11};
    \node[draw=none] at (c7.north) [yshift=1.5mm] {\scriptsize 12};

    \node[s] (d34) [below=of c34, yshift=-\rowspc+1em]  {$1$};
    \node[s] (d4)  [right=of d34]                   {$0$};
    \node[s] (d45) [right=of d4]                    {$1$};
    \node[s] (d5)  [right=of d45]                   {$0$};
    \node[s] (d56) [right=of d5]                    {$1$};
    \node[s] (d6)  [right=of d56]                   {$0$};
    \node[s] (d7)  [right=of d6]                    {$q$};
    \node[s] (d3)  [left=of  d34]                   {$0$};
    \node[s] (d23) [left=of  d3]                    {$1$};
    \node[s] (d2)  [left=of  d23]                   {$1$};
    \node[s] (d12) [left=of  d2]                    {$0$};
    \node[s] (d1)  [left=of  d12]                   {$0$};
    \node[s] (d0)  [left=of  d1]                    {$q$};
    \draw    (d0.north west) -- ++(-0.4cm,0)
             (d0.south west) -- ++(-0.4cm,0)
             (d7.north east) -- ++(0.4cm,0)
             (d7.south east) -- ++(0.4cm,0);
    \node[draw=none] at (d0.north) [yshift=1.5mm] {\scriptsize -1};
    \node[draw=none] at (d1.north) [yshift=1.5mm] {\scriptsize 0};
    \node[draw=none] at (d12.north) [yshift=1.5mm] {\scriptsize 1};
    \node[draw=none] at (d2.north) [yshift=1.5mm] {\scriptsize 2};
    \node[draw=none] at (d23.north) [yshift=1.5mm] {\scriptsize 3};
    \node[draw=none] at (d3.north) [yshift=1.5mm] {\scriptsize 4};
    \node[draw=none] at (d34.north) [yshift=1.5mm] {\scriptsize 5};
    \node[draw=none] at (d4.north) [yshift=1.5mm] {\scriptsize 6};
    \node[draw=none] at (d45.north) [yshift=1.5mm] {\scriptsize 7};
    \node[draw=none] at (d5.north) [yshift=1.5mm] {\scriptsize 8};
    \node[draw=none] at (d56.north) [yshift=1.5mm] {\scriptsize 9};
    \node[draw=none] at (d6.north) [yshift=1.5mm] {\scriptsize 10};
    \node[draw=none] at (d7.north) [yshift=1.5mm] {\scriptsize 11};

    \node (alphac) [draw=none, left=of c0, xshift=-1cm] {$\alpha(c)$};
    \node (c)      [draw=none, at=(alphac |- z0)]       {$c$};
    \node (deltac) [draw=none, at=(alphac |- d0)]       {$\Delta^\Exp(c)$};
  \end{tikzpicture}
  \caption{Illustration of a step of the XCA $A$.
    The number next to each cell indicates its index in the respective
    configuration.
  }
  \label{fig_ex_XCA}
\end{figure*}

The supply of hidden cells is never depleted; whenever a hidden cell becomes
visible, new hidden cells appear between it and its neighbors.
Thus, the number of active cells in an XCA may increase exponentially:

\begin{lemma}
  Let $A$ be an XCA.
  For an input of size $n$, $A$ has at most $a(t) = (n + 3)2^t - 3$ active cells
after $t \in \Nz$ steps.
  This upper bound is sharp.
  \label{lem_xca_number_cells}
\end{lemma}

\begin{proof}
  The claim is proven using induction on $t \in \Nz$.
The induction basis is evident since $A$ has exactly $n$ active cells in time
step $t = 0$.
  For the induction step, assume the claim holds for some $t \in \Nz$.
  Without loss of generality, it may also be assumed that the number of active
cells in $A$ is maximal (i.e., equal to $a(t)$).
  Then $A$ can have at most $2a(t) + 3$ active cells in time step $t + 1$ since
$a(t)$ many cells were already active, and a maximum of two quiescent and $a(t)
+ 1$ hidden cells may become active in the transition to the next step.
  The proof is complete by using $2a(t) + 3 = (n+3)2^{t+1} -3 = a(t+1)$.
  \qed
\end{proof}

We have postponed defining the acceptance condition for XCAs until now.
Usually, a CA possesses a distinguished cell, often cell $0$, which dictates
the automaton's accept or reject response \autocite{kutrib_language_theory}.
In the case of XCAs, however, under a reasonable complexity-theoretical
assumption (i.e., $\PTIME \neq \ttpNP$) such an acceptance condition results in
XCAs not making full use of the efficient cell growth indicated in
Lemma~\ref{lem_xca_number_cells} (see Section~\ref{sec_1XCA}).
This phenomenon does not occur if the acceptance condition is defined based on
\emph{unanimity}, that is, in order for an XCA to accept (or reject), \emph{all}
its cells must accept (or reject) simultaneously.
This acceptance condition is by no means novel \autocite{Ibarra_ACA,
sommerhalder_ACA, rosenfeld_book, kim_ACA}.
As an aside, note all (reasonable) CA time complexity classes (including, in
particular, linear- and polynomial-time) remain invariant when using this
acceptance condition instead of the standard one.

Also of note is that, for the standard acceptance condition, we insist on
\emph{unique} accept and reject states.
This serves to not only simplify some arguments in Section~\ref{sec_XCAP} but
also to show that unique states already suffice to decide problems in $\ttpNP$.
We revisit this topic in Section~\ref{sec_MAR_XCA}, where we consider XCAs with
multiple accept and reject states (and prove that the class of problems that can
be decided efficiently remains the same).

\begin{definition}[Acceptance condition, time complexity]
  \label{def_xca_acc_beh}
  Each XCA has a unique \emph{accept state} $a$ and a unique \emph{reject state}
$r$.
  An XCA $A$ halts if all active (and visible) cells are either all in state
$a$, in which case the XCA \emph{accepts}, or they are all in state $r$, in
which case it \emph{rejects}; if neither is the case, the computation continues.
  $L(A)$ denotes the set of words accepted by $A$.

  The \emph{time complexity} of an XCA (for an input $w$) is the number of
elapsed steps until it halts.
  An \emph{XCA decider} is an XCA which halts on every input.
  A language $L$ is in $\XCAP$ if there is an XCA decider $A'$ with polynomial
time complexity (in the length $|w|$ of $w$) and such that $L = L(A')$.
\end{definition}

In summary, the decision result of an XCA decider is the one indicated by the
first configuration in which its active cells are either all in the accept or
all in the reject state.
This agrees with our aforementioned notion of a \emph{unanimous} decision.

\section{Characterizing $\XCAP$}
\label{sec_XCAP}

This section covers the main result of this paper, that is, characterizing
$\XCAP$ as being equal to $\ttpNP$ (Theorem~\ref{thm_XCAP_ttpNP}).
It is subdivided into three parts:
First, we address a result from \textcite{BA} which is relevant towards proving
the aforementioned characterization.
Next, we state and prove Theorem~\ref{thm_XCAP_ttpNP}.
Finally, we discuss an alternative characterization of $\XCAP$ based on NTMs.

\subsection{An XCA for $\TAUT$}
\label{sec_XCAP_BA}

In this section, we cover the following result from \textcite{BA}, which
provides the starting point towards proving Theorem~\ref{thm_XCAP_ttpNP}:

\begin{proposition}
  $\NP \cup \coNP \subseteq \XCAP$.
  \label{prop_NP_coNP_in_XCAP}
\end{proposition}

Since many-one reductions by TMs can be simulated by (X)CAs in real-time, it
suffices to show $\XCAP$ contains $\NP$- and $\coNP$-complete problems.
We construct XCAs for $\SAT$ and $\TAUT$ which run in polynomial time and apply
Theorem~\ref{thm_sat_np_complete}.
Since acceptance and rejection are defined symmetrically (as per
Definition~\ref{def_xca_acc_beh}), if $L$ can be accepted by an XCA $A$, then
swapping the accept and reject states of $A$ we obtain an XCA that decides the
complement of $L$ with the exact same complexity as $A$.
Hence, as $\NP$ and $\coNP$ are complementary classes, it suffices to show
$\coNP \subseteq \XCAP$.

The key idea towards the result is that an XCA can efficiently duplicate
portions of its configuration:
Let a \emph{block} denote a subconfiguration $\# w \#$ where $\#$ is a
(special) separator symbol and $w$ is a word not containing $\#$.
In particular, starting from any such block, the XCA can, in a single step,
produce the block $\# w_2 \$ \#$ where $\$$ is a separator symbol different from
$\#$ and
\[
  w_2 = (w(0),0) (w(0),1) \cdots (w(|w| - 1),0) (w(|w| - 1),1)
\]
duplicates the word $w$ (i.e., we have $|w_2| = 2|w|$).
Using a stable sorting algorithm (following, e.g., techniques from
\textcite{CA_linear_sorting}), the XCA then sorts the symbols into place
according to the second component of the tuples above and obtains the
subconfiguration $\# w \$ w \#$ in a linear number of steps.

\begin{proof}
In accordance with the previous discussion, we construct an XCA $A$ for $\TAUT$.

Firstly, $A$ verifies the input $f$ is a syntactically correct formula; this can
be done, for instance, simply by simulating a TM for this task.
Following that, the operation of $A$ can be subdivided into two large steps.
In the first of them, $A$ iteratively expands its configuration (in a way we
shall describe in more detail) so as to cover every possible truth assignment
of the variables of $f$ and arrives at a configuration $c_f$ (detailed below).
The second step starts from $c_f$, computes the evaluations of $f$ under the
respective truth assignments in parallel, and accepts or rejects according to
whether the results are all \enquote{true} or not.
Both procedures require time polynomial in the length $|f|$ of $f$; thus, $A$
runs in polynomial time.

\paragraph{Step 1}

Given a Boolean formula $f$ over $m$ variables, let $x_0, \dots, x_{m-1}$ be the
ordering of its variables according to their first appearance in $f$ when this
is read from left to right.
Furthermore, letting $s_0, \dots, s_{2^m - 1}$ be the lexicographic ordering of
the strings in $\{ F, T \}^m$ (under the convention that $F$ precedes $T$), we
obtain a natural ordering $I_0, \dots, I_{2^m - 1}$ of the $2^m$ possible
interpretations of the $x_i$ by identifying each $s_j$ with the interpretation
$I_j$ that satisfies $I_j(x_i) = s_j(i)$.
We let $c_f$ be the following configuration:
\[
    \cdots \; q \; \# \; (E_0)^{|f|} \; \# \; \cdots
      \; \# \; (E_{2^m - 1})^{|f|} \; \# \; q \; \cdots
\]
where $E_j = E_{I_j}(f)$ is the evaluation of $f$ under $I_j$ (as a symbol,
i.e., an element of the alphabet $\{ F, T \}$).

We now further specify how $c_f$ is reached.
$A$ starts by surrounding $f$ with $\#$ delimiters.
Each block $\# f \#$ of $A$ repeats the following procedure as long as $f$
contains at least one variable:
\begin{enumerate}
  \item Duplicate $f$ (as described previously), yielding the subconfiguration
$\# f \$ f' \#$.
  \item Determine the first variable $v = xy$ in $f$, where $y \in \{ 0,1
\}^+$, and replace every occurrence of $v$ in $f$ (resp., $f'$) with $F^k$
(resp., $T^k$), where $k = |v| = 1 + |y|$.
  \item Replace the middle delimiter $\$$ with $\#$ and synchronize the two
blocks corresponding to $f$ and $f'$ (so they continue their operation at the
same time).
\end{enumerate}
When $f$ no longer contains a variable, the block evaluates it directly
(e.g., by simulating a TM for this task).

The correctness of the above is shown by induction on $m$.
The case $m = 0$ is trivial, so assume $m > 0$.
The above procedure replaces the variables of $f$ such that precisely $2^m$
copies are produced, each corresponding to an $I_j$ (and according to the
ordering described above).
Also note the blocks of $A$ always have the same length and, because of step
$3$ above and using transitivity, any two blocks are synchronized with each
other.
Thus, when $f$ has no variables left, the evaluations all happen and
terminate at the same time, thus producing the desired configuration.
Finally, it is straightforward to show the above procedure requires polynomial
time.

\paragraph{Step 2}

We now describe the second procedure of $A$ which, starting from the
configuration $c_f$ of, leads $A$ to accept or reject depending on the results
present in $c_f$.

Notice that, in the first step above, we ensured that $c_f$ is reached in such
a way that the blocks corresponding to the $E_i$ are all synchronized.
Hence, from this point each block (including the delimiting $\#$ cells)
initiates a synchronization, following which all cells in the block
simultaneously enter the accept (resp., reject) state if the respective
evaluation's result was $T$ (resp., $F$).
The reject state is maintained while accept states yield to a reject state,
that is, we have $\delta(q_1, r, q_2) = r$ and $\delta(q_1, a, q_2) = r$ for the
local transition function $\delta$ of $A$ and arbitrary states $q_1$ and $q_2$.
Thus, if all evaluations are \enquote{true} (i.e., their result is $T$), the
cells all simultaneously enter the accept state; otherwise, all cells
necessarily enter the reject state.
Since this process also takes only polynomial time, the claim follows.
\qed  
\end{proof}

We conclude this section by stressing that step 2 above builds on a
\enquote{trick} that is only possible due to the unanimous acceptance
condition of XCAs.
Assume, for the moment, that the XCA $A$ of before can continue computing (i.e.,
does not halt) even if it has reached a configuration in which it accepts or
rejects.
Then $A$ is guaranteed to eventually reach a configuration in which it rejects
regardless of what the results for the evaluations of $f$ are.
This means the only case in which $A$ is prevented from rejecting is when it
accepts (namely when all of the $E_i$ are \enquote{true}); that is, $A$ is
capable of \emph{rejecting under the condition it has not accepted}.
This kind of behavior is quite different from, say, an NTM (seen as an
alternating Turing machine with only existential states), where the result of
each computation branch is completely independent of the other branches.
We shall come back to this point later in Section~\ref{sec_syncNTM} and address
it from another perspective.

\subsection{A First Characterization}
\label{sec_XCAP_ttpNP}

In this section, we prove the main result of this paper:

\begin{theorem}
  \label{thm_XCAP_ttpNP}
  $\XCAP = \ttpNP$.
\end{theorem}

The equality in Theorem~\ref{thm_XCAP_ttpNP} is proven by considering the two
inclusions (Propositions~\ref{prop_ttpNP_in_XCAP} and \ref{prop_XCAP_in_ttpNP}).

\begin{proposition}
  $\ttpNP \subseteq \XCAP$.
  \label{prop_ttpNP_in_XCAP}
\end{proposition}

\begin{proof}
  The claim is shown by constructing an XCA $A$ that decides $\SATTAUT$ (see
Definition~\ref{def_sattaut} and Theorem~\ref{thm_tautsat_ttpNP}) in polynomial
time.
  The actual inclusion follows from the fact that CAs can simulate
polynomial-time many-one reductions by TMs in real-time.

  Given a problem instance $f$, $A$ evaluates $f$ recursively.
  Without loss of generality, we may assume $f = \land(f_1 \inalt_\SAT,
\lor(f_2 \inalt_\TAUT, f'))$, where $f'$ is a further problem instance; other
instances of $\SATTAUT$ are obtained by replacing $f_1$, $f_2$, or $f'$ with a
trivial formula (e.g., a trivial tautology).

  To evaluate $f_1 \inalt_\SAT$, $A$ emulates the behavior of the XCA for
$\SAT$ (see Proposition~\ref{prop_NP_coNP_in_XCAP}); however, special care must
be taken to ensure $A$ does not halt prematurely.
  All computation branches retain a copy of $f$.
  Whenever a branch obtains a \enquote{true} result, the respective cells do not
directly accept (as in the original construction); instead, they proceed with
evaluating the formula's next connective.
  Conversely, if the result is false, the respective cells simply enter the
reject state.
  The behavior for $f_2 \inalt_\TAUT$ is analogous, with $A$ emulating the XCA
for $\TAUT$ instead (and with exchanged accept and reject states, accordingly).
  Additionally, we require $\delta(q_1, r, q_2) = a$ and $\delta(q_1, a, q_2) =
r$ for every states $q_1$ and $q_2$, that is, once a cell enters the accept or
reject state, it is forced to unconditionally alternate between the two.
  To ensure $A$ is still able to accept or reject, we (arbitrarily) enforce
accept states only exist in even-numbered steps and reject states only in
odd-numbered ones.\footnote{%
  For example, have each cell contain a bit counter and, if needed, wait for one
step before transitioning to an accept or reject state.
}

  If $f_1 \not\in \SAT$, all branches of $A$ transition into the reject state,
and $A$ rejects.
  Otherwise, $f_1$ is satisfiable; thus, at least one branch obtains a
\enquote{true} result, and $A$ continues to evaluate $f$ until the
(aforementioned) base case is reached.
  An analogous argument applies for $f_2$.
  Note the synchronicity of the branches guarantee they operate exactly the same
and terminate at the same time.
  The repeated transition between accept and reject states guarantee the only
cells relevant for the final decision of $A$ are those in the branches which are
still \enquote{active} (in the sense they are still evaluating $f$).

  In conclusion, $A$ accepts $f$ if and only if it evaluates to true; otherwise,
$A$ rejects $f$.
  $A$ runs in polynomial time since $f$ has at most $|f|$ predicates and since
evaluating a predicate requires polynomial time in $|f|$.
\qed
\end{proof}

For the converse, we express an XCA computation as a $\SATTAUT$ instance.
The main effort here lies in defining the appropriate \enquote{variables}:

\begin{definition}[$\STATEALL$]
  Let $A$ be an XCA, and let $V_A$ be the set of triples $(w, t, z)$, $w$ being
an input for $A$, $t \in \{0, 1\}^+$ a (standard) binary representation of $\tau
\in \Nz$, and $z$ a state of $A$.
  $\STATEALL(A) \subseteq V_A$ is the subset of triples such that, if $A$ is
given $w$ as input, then after $\tau$ steps all active cells are in state $z$.
  \label{def_stateall}
\end{definition}

\begin{lemma}
  For any XCA $A$ with polynomial time complexity, $\STATEALL(A) \in \coNP$.
  \label{lem_stateall}
\end{lemma}

\begin{proof}
  Let $p\colon \N \to \Nz$ be a polynomial bounding the time complexity of $A$,
that is, for an input of size $n$, $A$ always terminates after at most $p(n)$
many steps.
  Suppose there is an NTM $T$ which \emph{covers} all active cells in step
$\tau$ of $A$ for input $w$, that is, for each such active cell $r$ there is at
least one computation branch of $T$ corresponding to $r$.
  Furthermore, assume $T$ can then compute the state $z'$ of $r$ in polynomial
time and accepts if and only if $z' = z$.
  Without restriction, we may assume $\tau \le p(|w|)$; this can be enforced by
$T$, for instance, by computing $p(|w|)$ and rejecting whenever $\tau > p(|w|)$.
  Then, the claim follows immediately from the existence of $T$:
  If all computation branches of $T$ accept, then in step $\tau$ all cells of
$A$ are in state $z$; otherwise, there is a cell in a state which is not $z$,
and $T$ rejects.

  The rest of the proof is concerned with the construction of a $T$ with the
properties just described.
  First, we describe the construction, followed by arguing it has the desired
complexity (which is fairly straightforward).
  The last part of the proof concerns the correctness of $T$ which, although
fairly evident, calls for a more technical argument.%
  \footnote{The main reason for this is that our construction of $T$
non-deterministically picks cells starting at the initial configuration of $A$
instead of (picking a final computation step and then) an arbitrary cell from
the final configuration.
    The issue with the latter approach is that then, in order to compute the
chosen cell's state, we would require a procedure that, given an arbitrary cell
$z$ in the final configuration of $A$ and \emph{without simulating $A$
directly}, determines whether $z$ was already present in the initial
configuration and, provided it was not, in which step exactly did it turn from a
hidden cell into an active one.
    This is indeed feasible if we constructed $A$ ourselves but virtually
impossible in case $A$ is arbitrary (which is the setting of the proof).}

  \paragraph{Construction}
  To compute the state of a (in particular, active) cell in step $\tau$, $T$
computes a series of \emph{subconfigurations} $c_0, \dots, c_\tau$ of $A$, that
is, contiguous excerpts of the global configuration of $A$.
  As the number of cells in an XCA may increase exponentially in the number
of computation steps, bounding $c_i$ is essential to ensure $T$ runs in
polynomial time; in particular, $T$ maintains $|c_i| = 1 + 2(\tau - i)$ (for
$i \ge 1$), thus ensuring the lengths of the $c_i$ are linear in $\tau$ (which,
in turn, is polynomial in $|w|$).
  This choice of length for the $c_i$ ensures each of the subconfigurations
correspond to a cell of $A$ surrounded by $\tau - i$ cells on either
side (i.e., each $c_i$ corresponds to the extended neighborhood of radius $\tau
- i$ of said cell).
  The non-determinism of $T$ is used exclusively in picking the cells from
$c_i$ which are to be included in the next subconfiguration $c_{i+1}$.

  The initial subconfiguration $c_0$ is set to be $q^{2\tau} w q^{2\tau}$,
thus containing the input word as well as (as shall be proven) a sufficiently
large number of surrounding quiescent cells.
  To obtain $c_{i+1}$ from $c_i$, $T$ applies the transition function of $A$ to
$c_i$ and obtains a new temporary subconfiguration $c_{i+1}'$.
  The next state of the two \enquote{boundary} cells (i.e, those belonging to
indices $0$ and $|c_i| - 1$) cannot be determined, and they are excluded from
$c_{i+1}'$.
  As a result, $c_{i+1}'$ contains $|c_i| - 2$ cells from the previous
configuration $c_i$, plus a maximum of $|c_i| - 1$ additional cells which were
previously hidden.
  Therefore, to maintain $|c_{i+1}| = 1 + 2(\tau - (i+1))$, $T$
non-deterministically sets $c_{i+1}$ to a contiguous subset of $c_{i+1}'$
containing exactly $1 + 2(\tau - (i+1)) \le |c_i| - 2$ cells.

  The process of selecting a next subconfiguration $c_{i+1}$ from $c_i$ is
depicted in Figure~\ref{fig_lem_stateall_subconfig_simulation}.
  In the illustration, $|c_i|$ has been replaced with $n$ for legibility.
  $T$ at first applies the global transition function of $A$ to obtain an
intermediate subconfiguration $c_{i+1}'$ with $m = |c_{i+1}'|$ cells.
  Because of hidden cells, $c_{i+1}'$ may consist of $n - 2 \le m \le 2n - 3$
cells.
  Non-determinism is used to select a contiguous subconfiguration of $n - 2$
cells, thus giving rise to $c_{i+1}$.

  \begin{figure}
    \centering
    \begin{tikzpicture}[every node/.style={block},
        block/.style={
          minimum width=5mm,
          minimum height=height("$z'_{n-1}$") + 3.5mm,
          outer sep=0pt,
          draw, rectangle, node distance=0pt}]

      \setlength{\rowspc}{1.35cm}

      \node (z1) {$z_1$};
      \node (z2) [right=of z1] {$z_2$};
      \node[minimum width=7em] (mid1) [right=of z2] {$\dots$};
      \node (zn-1) [right=of mid1] {$z_{n-1}$};
      \node (zn) [right=of zn-1] {$z_n$};

      \draw (z1.north west) -- ++(-0.5cm,0)
            (z1.south west) -- ++ (-0.5cm,0)
            (zn.north east) -- ++(0.5cm,0)
            (zn.south east) -- ++ (0.5cm,0);

      \node (ci) [left=of z1, draw=none, xshift=-2cm] {$c_i$};

      \node[minimum width=11em] (mid2) %
          [below=of mid1, yshift=-\rowspc] {$\dots$};
      \node (z'2) [left=of mid2] {$z_2'$};
      \node (z'1) [left=of z'2] {$z_1'$};
      \node (z'm-1) [right=of mid2] {$z_{m-1}'$};
      \node (z'm) [right=of z'm-1] {$z_m'$};

      \draw (z'1.north west) -- ++(-0.5cm,0)
            (z'1.south west) -- ++ (-0.5cm,0)
            (z'm.north east) -- ++(0.5cm,0)
            (z'm.south east) -- ++ (0.5cm,0);

      \node[draw=none] (ci+1') at (ci |- z'1) {$c_{i+1}'$};

      \draw[dashed] (z2.south west) -- (z'1.north west)
                    (zn-1.south east) -- (z'm.north east);

      \coordinate (boxsw) at ([xshift=.75mm, yshift=-1.5mm] z'2.south west);
      \coordinate (boxse) at ([xshift=3.5cm] boxsw);
      \coordinate (boxnw) at ([xshift=.75mm, yshift=1.5mm] z'2.north west);
      \coordinate (boxne) at ([xshift=3.5cm] boxnw);
      \draw (boxsw) -- (boxse) |- (boxnw) -- (boxsw);

      \draw[|-|] ([yshift=1.75mm] boxnw) -- ([yshift=1.75mm] boxne)
        node[midway, above, draw=none, yshift=-1mm] {\scriptsize $n - 2$};

      \node[minimum width=4em] (mid3) %
          [below=of mid2, yshift=-\rowspc] {$\dots$};
      \node (z''2) [left=of mid3] {$z_2''$};
      \node (z''1) [left=of z''2] {$z_1''$};
      \node (z''n-2) [right=of mid3] {$z_{n-3}''$};
      \node (z''n-1) [right=of z''n-2] {$z_{n-2}''$};

      \draw (z''1.north west) -- ++(-0.5cm,0)
            (z''1.south west) -- ++ (-0.5cm,0)
            (z''n-1.north east) -- ++(0.5cm,0)
            (z''n-1.south east) -- ++ (0.5cm,0);

      \node[draw=none] (ci+1) at (ci |- z''1) {$c_{i+1}$};

      \draw[dashed] (boxsw) -- (z''1.north west)
                    (boxse) -- (z''n-1.north east);
    \end{tikzpicture}
    \caption{Illustration of how $T$ obtains the next subconfiguration
$c_{i+1}$ from $c_i$.}
    \label{fig_lem_stateall_subconfig_simulation}
  \end{figure}
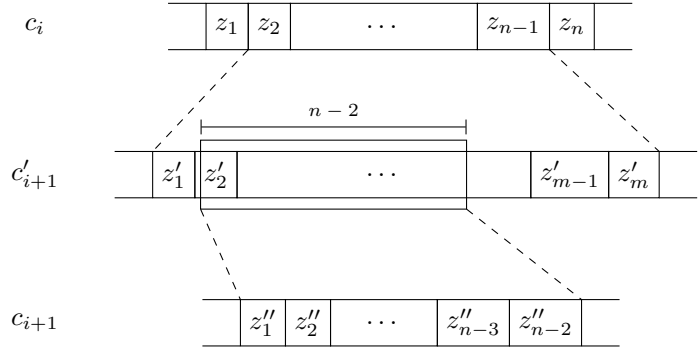

  \paragraph{Complexity}
  $T$ runs in polynomial time since the invariant $|c_i| = 1 + 2(\tau - i)$
guarantees the number of states $T$ computes in each step is bounded by a
multiple of $\tau$, which, in turn, we assumed to be bounded by $p(|w|)$.
  Only $|w|$ has to be taken into account when estimating the time complexity
of $T$ since the encoding of $z$ is $O(1)$ long, while that of $t$ has length
$O(\log p(|w|)) = O(\log |w|)$; as a result, the problem instance $(w, t, z)$
has length $O(|w|)$.

  \paragraph{Correctness}
  To show $T$ covers all active cells of $A$ in step $\tau$, it suffices to
prove the following by induction:
  Let $i \in \{ 0, \dots, \tau \}$, and let $z_1, \dots, z_m$ be the active
cells of $A$ in step $i$; then $T$ \emph{covers} all subconfigurations of
$q^{2(\tau - i)} z_1 \dotsm z_m q^{2(\tau - i)}$ of size $1 + 2(\tau - i)$,
that is, for every such subconfiguration $s$ there is a branch of $T$ in which
it picks $s$ as its $c_i$.
  Note this corresponds to $T$ covering all subconfigurations of $A$ in step
$i$ which contain at least one active cell; thus, when $T$ reaches step $\tau$,
it covers all subconfigurations of $z_1 \dotsm z_m$ of size $1$, that is, all
active cells.

  The induction basis follows from $c_0 = q^{2\tau} w q^{2\tau}$.
  For the induction step, fix a step $0 < i \le \tau$ and assume the claim
holds for all steps prior to $i$.
  To each subconfiguration of $q^{2(\tau - i)} z_1 \dotsm z_m q^{2(\tau - i)}$
having size $1 + 2(\tau - i)$ corresponds a cell $r$ which is located in its
center; thus, we may unambiguously denote every such subconfiguration by
$s_i(r)$.
  Now let $s_i(r)$ be given and consider the following three cases: $r$ was
active in step $i - 1$; $r$ was a hidden cell which became active in the
transition to step $i$; or $r$ was a quiescent cell in step $i - 1$ and, since
$|s_i(r)| = 1 + 2(\tau - i)$ and $r$ is the middle cell of $s_i(r)$, $r$ is at
most $\tau - i$ cells away from $z_1$ or $z_m$.

  In the first case, by the induction hypothesis, there is a value of
$c_{i-1}$ corresponding to $s_{i-1}(r)$; since only the two boundary cells are
present in $c_{i-1}$ but not in $c_i'$, $T$ can choose $c_i$ from $c_i'$ with
$r$ as its middle cell and obtain $s_i(r)$.
  In the second case, for any of the two parents $p_1$ and $p_2$ of $r$, there
are, by the induction hypothesis, values of $c_{i-1}$ which equal
$s_{i-1}(p_1)$ and $s_{i-1}(p_2)$; in either case, choosing $c_i$ from $c_i'$
with $r$ as its middle cell again yields $s_i(r)$.

  Finally, if $r$ was a quiescent cell, then, without loss of generality,
consider the case in which $r$ was located to the left of the active cells in
step $i - 1$.
  By the induction hypothesis, for each cell $r'$ up to $\tau - i + 1$ cells
away from the leftmost active cell $z_1$ there is a value of $c_{i-1}$
corresponding to $s_{i-1}(r')$, and the first case applies; the only exception
is if $c_i$ would then contain only quiescent cells, in which case $r$ would be
located strictly more than $\tau - i$ cells away from $z_1$, thus contradicting
our previous assumption.
  The claim follows.
\qed
\end{proof}

\begin{proposition}
  $\XCAP \subseteq \ttpNP$.
  \label{prop_XCAP_in_ttpNP}
\end{proposition}

\begin{proof}
    Let $L \in \XCAP$, and let $A$ be an XCA for $L$ whose time complexity is
bounded by a polynomial $p\colon \N \to \Nz$.
    Additionally, let $w$ be an input for $A$, $V_A$ be as in
Definition~\ref{def_stateall}, and let $V = V_A \cdot \{ \inalt_{\STATEALL(A)}
\}$, where $\inalt_{\STATEALL(A)}$ is a syntactic symbol standing for membership
in $\STATEALL(A)$ (cf.\ Definition~\ref{def_sattaut}).
    Define $f_0(w), \dots, f_{p(n)}(w) \in \BOOL_V$ recursively by
    \begin{align*}
      f_i(w) = \lor( \, &(w, i, a) \inalt_{\STATEALL(A)}, \\
              &\land( \, \lnot( \, (w, i, r) \inalt_{\STATEALL(A)} \, ),
                        f_{i+1}(w) \, ) \, )
    \end{align*}
    for $i < p(n)$ and
    \begin{align*}
      f_{p(n)}(w) = \lor( \, &(w, p(n), a) \inalt_{\STATEALL(A)}, \\
              &\lnot( \, (w, p(n), r) \inalt_{\STATEALL(A)} \, ) \, ).
    \end{align*}

    Lemma~\ref{lem_stateall} together with the $\coNP$-completeness of
$\TAUT$ (see Theorem~\ref{thm_sat_np_complete}) ensures each subformula
of the form $(w, i, a) \inalt_{\STATEALL(A)}$ is polynomial-time many-one
reducible to an equivalent (in the sense of evaluating to the same truth value
under the respective interpretations; see Definition~\ref{def_sattaut})
$\SATTAUT$ formula $g \inalt_\TAUT$, $g$ being a $\TAUT$ instance.
    Similarly, each subformula $\lnot((w, i, a) \inalt_{\STATEALL(A)})$
is reducible to an equivalent formula $h \inalt_\SAT$.
    Since each of the $f_i(w)$ may contain only polynomially (respective to
$|w|$) many connectives, each is polynomial-time (many-one) reducible to an
equivalent $\SATTAUT$ instance $f_i'(w)$.

    By the definition of XCA (i.e.,
Definitions~\ref{def_xca}~and~\ref{def_xca_acc_beh}) and our choice of
$p$, $f'(w) = f_0'(w)$ is true if and only if $A$ accepts $w$.
    Since $f'(w)$ is such that $|f'(w)|$ is polynomial in $|w|$, this provides
a polynomial-time (many-one) reduction of $L$ to a problem instance of
$\SATTAUT \in \ttpNP$.
    The claim follows.
\qed
\end{proof}

This concludes the proof of Theorem~\ref{thm_XCAP_ttpNP}.

\subsection{A Turing Machine Characterization}
\label{sec_syncNTM}

We now turn to a closer investigation of the relation between XCA
polynomial-time computations and the class $\ttpNP$.
In this section, we shall view NTMs as a special case of alternating Turing
machines (ATMs), that is, as possessing a computation tree in which all branches
are existential.
Recall the computational strategy of the XCA in
Proposition~\ref{prop_NP_coNP_in_XCAP} essentially consists of creating multiple
computation branches, each corresponding to a possible variable assignment of
the input formula.
In a sense, this merely replicates the standard NTM construction used to show
$\SAT \in \NP$ (or, equivalently, $\TAUT \in \coNP$).

Nevertheless, it is widely suspected that $\XCAP = \ttpNP$ is a strictly larger
class than $\NP$, and it is a fair point to question exactly why it is that we
obtain such a class (instead of merely $\NP$).
The explanation ultimately lies in the acceptance condition of XCAs.
Consider that, for instance, the presence of a non-accepting cell prevents
acceptance; thus, by the automaton not halting, would-be accepting branches are
made aware of the existence of this cell.
This enables a form of information transfer between computation branches which
is not possible in NTMs.
In fact, this form of interaction is not exclusive to a model based on CAs but,
as we shall see, may also be expressed in terms of a model based on Turing
machines.

In the following definition, we extract the essence of this interaction and
embed it into the NTM model.
The novelty consists in a modification to the acceptance condition, which, as
is the case for XCAs (see Definition~\ref{def_xca_acc_beh}), requires a
simultaneous decision across all computation branches.
Unsurprisingly, the condition is that of a unanimous decision across the
branches (instead of a single branch being accepting) and actually resembles
more a characterization of $\coNP$ than of $\NP$ (by an NTM variant which
accepts if and only if all non-deterministic branches are accepting or,
equivalently, an ATM possessing only universal states).
However, note this by no means deviates from our goal, that is, defining a model
based (exclusively) on TMs that features the form of information transfer
discussed above.

\begin{definition}[SimulNTM]
  A \emph{simultaneous NTM} (\emph{SimulNTM}) is an NTM $T$ having the property
that, for any input $w$ of $T$, there is $t \in \Nz$ such that, in step $t$,
the computation branches of $T$ are either all accepting or all rejecting.
  Furthermore, if $t$ is minimal with this property, then $T$ \emph{accepts}
(resp., \emph{rejects}) if all branches in step $t$ are accepting (resp.,
rejecting).
  $\SimulNP$ denotes the class of languages decided by SimulNTMs in polynomial
time.
\end{definition}

Refer to Figure~\ref{fig_SimulNTM_example} for an example illustrating the
computation of a SimulNTM $T$ with accept state $a$ and reject state $r$.
Upon reaching step number $t'$, $T$ does not yet terminate since some of the
computation branches are accepting while some are still rejecting, that is,
there is no unanimity.
$T$ accepts in step $t$ since then all its branches are in state $a$ (assuming
this was not the case in any step prior to step $t$).

\begin{figure}
  \centering
  \begin{tikzpicture}[->, >=stealth', shorten >=1pt, auto, semithick,
      node distance=0.65cm]
    \node[state, initial above, initial text={}] (a) {};

    \node[state, xshift=-0.75cm] (b0) [below left =of a] {};
    \node[state, xshift=0.75cm] (b1) [below right =of a] {};

    \node[state] (c0) [below =of b0] {};
    \node[state] (c2) [below =of b1] {};
    \node[state] (c1) [left =of c2] {};
    \node[state] (c3) [right =of c2] {};

    \node (d0) [below =of c0] {$\vdots$};
    \node (d1) [below =of c1] {$\vdots$};
    \node (d2) [below =of c2] {$\vdots$};
    \node (d3) [below =of c3] {$\vdots$};

    \node[state, fill=red] (e2) [below =of d1] {$r$};
    \node[state, fill=green] (e3) [below =of d2] {$a$};
    \node[state, fill=red] (e4) [below =of d3] {$r$};
    \node[state, fill=green] (e1) [left =of e2] {$a$};
    \node[state, fill=green] (e0) [left =of e1] {$a$};

    \node (t') [left =of e0] {Step $t'$};

    \node (f0) [below =of e0] {$\vdots$};
    \node (f1) [below =of e1] {$\vdots$};
    \node (f2) [below =of e2] {$\vdots$};
    \node (f3) [below =of e3] {$\vdots$};
    \node (f4) [below =of e4] {$\vdots$};

    \node[state, fill=green] (g0) [below =of f0] {$a$};
    \node[state, fill=green] (g1) [below =of f1] {$a$};
    \node[state, fill=green] (g2) [below =of f2] {$a$};
    \node[state, fill=green] (g3) [below =of f3] {$a$};
    \node[state, fill=green] (g4) [below =of f4] {$a$};

    \node (t) at (t' |- g0) {Step $t$};

    \node (t0) at (t |- a) {Step $0$};
    \node (t1) at (t |- b0) {Step $1$};
    \node (t2) at (t |- c0) {Step $2$};

    \path (a)  edge (b0)
          (a)  edge (b1)

          (b0) edge (c0)
          (b1) edge (c1)
          (b1) edge (c2)
          (b1) edge (c3)

          (c0) edge (d0)
          (c1) edge (d1)
          (c2) edge (d2)
          (c3) edge (d3)

          (d0) edge (e0)
          (d0) edge (e1)
          (d1) edge (e2)
          (d2) edge (e3)
          (d3) edge (e4)

          (e0) edge (f0)
          (e1) edge (f1)
          (e2) edge (f2)
          (e3) edge (f3)
          (e4) edge (f4)

          (f0) edge (g0)
          (f1) edge (g1)
          (f2) edge (g2)
          (f3) edge (g3)
          (f4) edge (g4)

          ;
  \end{tikzpicture}
  \caption{Illustration of the operation of a SimulNTM.
    States other than $a$ or $r$ have been omitted for simplicity.}
  \label{fig_SimulNTM_example}
\end{figure}
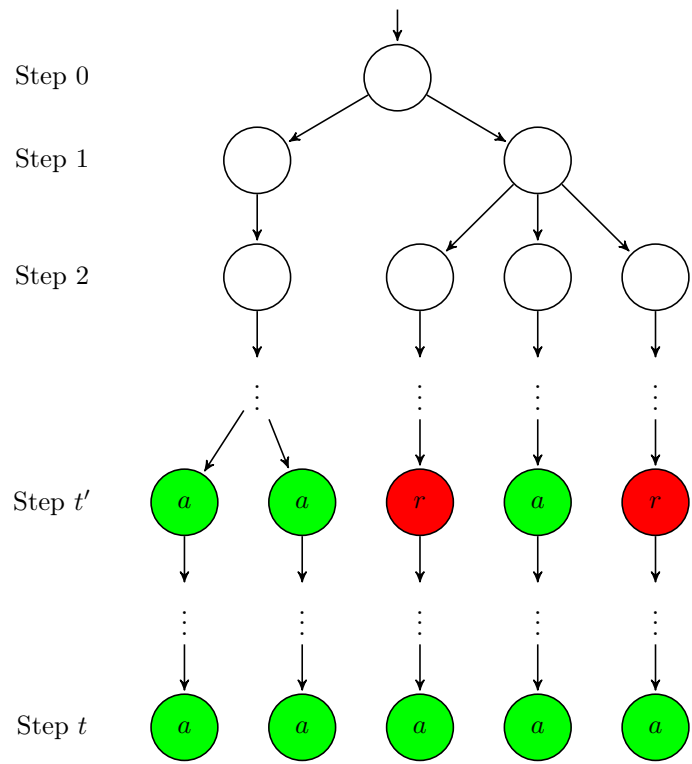

\begin{theorem}
  $\SimulNP = \XCAP = \ttpNP$.
  \label{thm_simulNP}
\end{theorem}

The proof uses techniques fairly similar to the previous ones in this section.

\begin{proof}
  The claim is shown by proving the two inclusions, both of which, in turn, are
proven by polynomial-time simulation of either model by the other one.

  For the inclusion $\SimulNP \subseteq \XCAP$, let $T$ be a SimulNTM whose
running time is bounded by a polynomial $p\colon \N \to \Nz$, to which we shall
construct a polynomial-time XCA $A$ with $L(A) = L(T)$.
  Strictly speaking, $A$ is not as in Definition~\ref{def_xca_acc_beh} since
it has multiple accept and reject states (i.e., $A$ is an MAR-XCA; see section
\ref{sec_MAR_XCA}); as mentioned in Section~\ref{sec_CA} and proven in
Theorem~\ref{thm_marxca}, however, this is equivalent to the original definition
(i.e., Definition~\ref{def_xca_acc_beh}).
  As is the case for ATMs, we may assume $T$ always creates one additional
branch in each step, that is, if its computation is viewed as a tree, then each
node has outdegree precisely $2$.

  $A$ maintains a separate \emph{block} of cells for each branch in the
computation of $T$.
  Each block contains the respective instantaneous configuration of $T$ and is
updated according to the rules of $T$.
  The simulation of $T$ is advanced every $m = m(b)$ steps, where $b$ denotes
the current length of the respective block and $m$ we shall yet specify.
  After each simulated step of $T$, one blank symbol is created on either end
of the represented configuration; this is so that $T$ has (theoretically)
unbounded space while ensuring any two blocks always have the same length.
  When the computation of $T$ creates an additional branch, the respective block
creates a copy of itself and updates it so as to reflect the instantaneous
configuration of the new branch (parallel to updating its own configuration).
  Additionally, if the head of $T$ becomes accepting or rejecting, the cell
representing it sends signals to the other cells in the block so that they
\emph{mark} themselves as such accordingly.
  Here, \enquote{mark} means the respective cell changes into a state in which
it behaves exactly the same way as before (i.e., as if it was not marked), only
this state is an accepting or rejecting state (as determined by the respective
state of $T$).
  Once all cells in a block have marked themselves, they wait for an additional
step (so that $A$ may possibly accept or reject), after which all cells in the
block are unmarked again.

  We now set $m$ to be the total number of steps required by the two
aforementioned procedures, that is, creating a new branch and (if applicable)
marking cells as accepting or rejecting and subsequently unmarking them.
  Note that $m \in \Theta(b)$ is computable in real-time (by a block) as a
function of $b$.
  As an aside, also note the entire procedure described above does not require
any synchronization between the blocks whatsoever since it consists solely of
operations that each require a fixed number of steps and, in addition, the
simulation is advanced every $m$ steps, which is also fixed.

  If the branches of $T$ are all accepting at the same time, then so are all
cells of $A$ (at the respective simulation step).
  The converse also holds:
  If the branches of $T$ all reject at the same time, then so do the cells of
$A$.
  In addition, because $m \in \Theta(b)$ and $b \in O(p(n))$ for an input of
length $n$, the running time of $A$ is polynomial, and $\SimulNP \subseteq
\XCAP$ follows.

  To prove the converse inclusion, given an XCA $A$ with running time bounded by
a polynomial $p\colon \N \to \Nz$, we construct a polynomial-time SimulNTM $T$
with $L(T) = L(A)$.
  Given an input $w$ for $A$, $T$ first sets $t = 0$ and then executes the
following procedure:
\begin{enumerate}
  \item Branch over all active cells of $A$ in time step $t$ (using, e.g., the
non-deterministic procedure described in the proof of Lemma~\ref{lem_stateall})
and compute the state $z$ of the cell that was chosen.
  \item If $z$ is the accept state of $A$, assume an accepting state for exactly
one computation step and then a non-rejecting state for exactly one step.
  If $z$ is the reject state of $A$, assume a non-accepting state followed by a
rejecting state.
  If $z$ is the quiescent state, assume an accepting state followed by a
rejecting state.
  If none of the cases above hold (i.e., $z$ is an active state that is neither
the accept nor the reject state), wait for two steps in a state that is neither
accepting nor rejecting.
  \item Increment $t$ and repeat.
\end{enumerate}

  Since the lengths of the configurations $c_i$ (see the proof of
Lemma~\ref{lem_stateall}) are the same regardless of how they are chosen, the
branches of $T$ can all be synchronized in their computation of the $c_i$ so
that they advance the simulation of the respective cell block at the same time
and, therefore, arrive at the respective state $z$ simultaneously.
  The subsequent instruction ensures $L(T) = L(A)$ since, if $A$ accepts (resp.,
rejects) its input in step $\tau$, then so do all branches of $T$ accept
(resp., reject) simultaneously for $t = \tau$, and the converse also holds.
  In addition, note that, by definition of $\STATEALL(A)$, $T$ is guaranteed to
halt since $(w, \tau, z) \in \STATEALL(A)$ must hold for some $\tau \le p(|w|)$
and $z \in \{ a, r \}$.
  Since $T$ is only slower than the NTM in the proof of Lemma~\ref{lem_stateall}
by a factor $O(p(|w|))$, it also runs in polynomial time.
  \qed
\end{proof}

\section{Immediate Implications}
\label{sec_implications}

This section covers some immediate corollaries of Theorem~\ref{thm_XCAP_ttpNP}
regarding XCA variants.
In particular, we address XCAs with multiple accept and reject states, followed
by XCAs with acceptance conditions differing from that in
Definition~\ref{def_xca_acc_beh}, in particular the two other classical
acceptance conditions for CAs \autocite{rosenfeld_book}.

\subsection{XCAs with Multiple Accept and Reject States}
\label{sec_MAR_XCA}

Recall the definition of an XCA specifies a single accept and a single reject
state (see Section~\ref{sec_CA}).
Consider XCAs with multiple accept and reject states.
As shall be proven, the respective polynomial-time class ($\MARXCAP$) remains
equal to $\XCAP$.
In the case of TMs, the equivalent result (i.e., TMs with a single accept and
a single reject state are as efficient as standard TMs) is trivial, but such is
not the case for XCAs.
Recall the acceptance condition of an XCA requires orchestrating the states of
multiple, possibly exponentially many cells.
In addition, an XCA with multiple accept states may, for instance, attempt to
accept whilst saving its current state (i.e., a cell in state $z$ may assume an
accept state $a_z$ while simultaneously saving state $z$).
Such is not the case for standard XCAs (i.e., as specified in
Definition~\ref{def_xca_acc_beh}), in which all accepting cells have
necessarily the same state.

\begin{definition}[MAR-XCA]
  A \emph{multiple accept-reject XCA} (\emph{MAR-XCA}) $A$ is an XCA with state
set $Q$ and which admits subsets $\Qacc, \Qrej \subseteq Q$ of accept and
reject states, respectively.
  $A$ accepts (resp., rejects) if its active cells all have states in $\Qacc$
(resp., $\Qrej$), and it halts upon accepting or rejecting.
  In addition, $A$ is required to either accept or reject its input after a
finite number of steps.
  $\MARXCAP$ denotes the MAR-XCA analogue of $\XCAP$.
\end{definition}

The following generalizes $\STATEALL$ (see Definition~\ref{def_stateall} and
Lemma~\ref{lem_stateall}) to the case of MAR-XCAs:

\begin{definition}[$\MARSTATEALL$]
  Let $A$ be an MAR-XCA with state set $Q$, and let $V_A$ be the set of triples
$(w, t, Z)$, $w$ being an input for $A$, $t \in \{0, 1\}^+$ a binary encoding of
$\tau \in \Nz$, and $Z \subseteq Q$.
  $\MARSTATEALL(A) \subseteq V_A$ is the subset of triples such that, if $A$ is
given $w$ as input, after $t$ steps all active cells have states in $Z$.
  \label{def_marstateall}
\end{definition}

\begin{lemma}
  For any MAR-XCA $A$ with polynomial time complexity, $\MARSTATEALL(A) \in
\coNP$.
  \label{lem_marstateall}
\end{lemma}
\begin{proof}
    Adapt the NTM from the proof of Lemma~\ref{lem_stateall} so as to accept if
and only if the last state is contained in $Z$.
    \qed
\end{proof}

Proceeding as in the proof of Proposition~\ref{prop_XCAP_in_ttpNP} (simply using
$\MARSTATEALL$ instead of $\STATEALL$) yields:
\begin{theorem}
  $\MARXCAP = \XCAP$.
  \label{thm_marxca}
\end{theorem}
\begin{proof}
    Define formulas $f_i(w)$ as in the proof of
Proposition~\ref{prop_XCAP_in_ttpNP} while replacing $\STATEALL$ with
$\MARSTATEALL$, the accept state $a$ with the set $\Qacc$, and the reject state
$r$ with the set $\Qrej$.
    Lemma~\ref{lem_marstateall} guarantees the reductions to $\SATTAUT$ are all
efficient.
    Thus, $\MARXCAP \subseteq \ttpNP = \XCAP$.
    Since MAR-XCAs are a generalization of XCAs, the converse inclusion is
trivial.
    \qed
\end{proof}

\subsection{Existential XCA}

The remainder of this section is concerned with XCAs variants which use
the two other classical acceptance conditions for CAs \autocite{rosenfeld_book}.
The first is that of a single final state being present in the CA's
configuration sufficing for termination.
We use the term \emph{existential} as an allusion to the existential states of
ATMs.

\begin{definition}[EXCA]
  An \emph{existential XCA} (EXCA) is an XCA with the following acceptance
condition:
  If at least one of its cells is in the accept (resp., reject) state $a$
(resp., $r$), then the EXCA accepts (resp., rejects).
  The coexistence of accept and reject states in the same global configuration
is disallowed (and any machine contradicting this requirement is, by
definition, \emph{not} an EXCA).
  $\EXCAP$ denotes the EXCA analogue of $\XCAP$.
\end{definition}

Disallowing the coexistence of accept and reject states in the global
configuration of an EXCA is necessary to ensure a consistent condition for
acceptance.
An alternative would be to establish a priority relation between the two (e.g.,
an accept state overrules a reject one); nevertheless, this behavior can be
emulated by our chosen variant with only constant delay.
This is accomplished by introducing binary counters to delay state transitions
and assure, for instance, that accept and reject states exist only in
even- and odd-numbered steps, respectively.

\begin{theorem}
  $\EXCAP = \XCAP = \ttpNP$.
  \label{thm_EXCAP_in_ttpNP}
\end{theorem}

Note this is an equivalence between two disparately complex acceptance
conditions:
As specified in Definition~\ref{def_xca_acc_beh}, all cells of an XCA must
agree on the final decision; on the other hand, in an EXCA, a single, arbitrary
cell suffices.
We ascribe this phenomenon to $\XCAP = \ttpNP$ being equal to its complementary
class.

As for the proof of Theorem~\ref{thm_EXCAP_in_ttpNP}, first note that
Proposition~\ref{prop_ttpNP_in_XCAP} may easily be restated in the context of
EXCAs:
\begin{proposition}
  $\ttpNP \subseteq \EXCAP$.
  \label{prop_ttpNP_in_EXCAP}
\end{proposition}

\begin{proof}
  By adapting the XCA $A$ for $\SATTAUT$ from the proof of
Proposition~\ref{prop_ttpNP_in_XCAP}, we obtain a polynomial-time EXCA $B$ for
$\TAUTSAT$.
  Here, $\TAUTSAT$ is the problem analogous to $\SATTAUT$ and which is obtained
simply by exchanging \enquote{$\TAUT$} and \enquote{$\SAT$} in
Definition~\ref{def_sattaut}.
  As $\SATTAUT$, it is straightforward to show $\TAUTSAT$ is $\ttpNP$-complete
(see also Theorem~\ref{thm_tautsat_ttpNP}).

  To evaluate a predicate of the form $f \inalt_\TAUT$, $B$ proceeds as $A$ and
emulates the behavior of the XCA deciding $\TAUT$ (see
Proposition~\ref{prop_NP_coNP_in_XCAP}); however, unlike $A$, the computation
branches of $B$ which evaluate to false reject immediately while it is those
that evaluate to true that continue evaluating the input formula.
  As a result, if $f \in \TAUT$, all branches of $B$ evaluate to true and
continue evaluating the input in a synchronous manner; otherwise, there is a
branch evaluating to false, and, since a single rejecting cell suffices for it
to reject, $B$ rejects immediately.
  The evaluation of $f \inalt_\SAT$ is carried out analogously.

  The modifications to $A$ to obtain $B$ do not impact its time complexity
whatsoever; thus, $B$ also has polynomial time complexity.
  \qed
\end{proof}

For the converse inclusion, consider the following $\NP$ analogue of the
$\STATEALL$ language (cf.\ Definition~\ref{def_stateall} and
Lemma~\ref{lem_stateall}):
\begin{definition}[$\STATEONE$]
  Let $A$ be an XCA and $V$ be the set of triples $(w, t, z)$ as in
Definition~\ref{def_stateall}.
  $\STATEONE \subseteq V$ is the subset of triples such that, for the input $w$,
after $t$ steps \emph{at least one} of the active cells of $A$ is in state $z$.
  \label{def_stateone}
\end{definition}

\begin{lemma}
  For any XCA $A$ with polynomial time complexity, $\STATEONE(A) \in \NP$.
  \label{lem_stateallo}
\end{lemma}
\begin{proof}
    Consider the NTM $T$ from Lemma~\ref{lem_stateall} and notice that, if any
of the active cells of $A$ in step $\tau$ have state $z$, then $T$ will have at
least one accepting branch; otherwise, none of the active cells of $A$ in step
$\tau$ have state $z$; thus, all branches of $T$ are rejecting.
    \qed
\end{proof}

Using Lemma~\ref{lem_stateallo} to proceed as in
Proposition~\ref{prop_XCAP_in_ttpNP} yields the following, from which
Theorem~\ref{thm_EXCAP_in_ttpNP} follows:
\begin{proposition}
  $\EXCAP \subseteq \ttpNP$.
  \label{prop_EXCAP_in_ttpNP}
\end{proposition}

\subsection{One-Cell-Decision XCA}
\label{sec_1XCA}

We turn to the discussion of XCAs whose acceptance condition is defined in terms
of a distinguished cell which directs the automaton's decision, considered the
standard acceptance condition for CAs \autocite{kutrib_language_theory}.
This condition is similar to the existential variant in the sense that the
automaton's termination is triggered by a single cell entering a final state.
The difference is that, here, the position of this cell is fixed.

We consider only the case in which the decision cell is the leftmost active cell
in the initial configuration (i.e., cell $0$).
By a \emph{one-cell-decision XCA} (\emph{1XCA}) we refer to an XCA which
accepts if and only if $0$ is in the accept state and rejects if and only if
cell zero is in the reject state.
Let $\oXCAP$ denote the polynomial-time class of 1XCAs.

The position of the decision cell is fixed; with a polynomial-time restriction
in place, it can only communicate with cells which are a polynomial (in
the length of the input) number of steps apart.
As a result, despite a 1XCA being able to efficiently increase its number
of active cells exponentially (see Lemma~\ref{lem_xca_number_cells}), any cells
impacting its decision must be at most a polynomial number of cells away from
the decision cell.
Thus:
\begin{theorem}
  $\oXCAP = \PTIME$.
  \label{thm_1XCAP_P}
\end{theorem}
\begin{proof}
    The inclusion $\oXCAP \supseteq \PTIME$ is trivial.
    For the converse, recall the construction of the NTM $T$ in
Lemma~\ref{lem_stateall}.
    $T$ can be modified so that it works deterministically and always chooses
the next configuration $c_{i+1}$ from $c_i$ by selecting cell zero as the
middle cell.
    If cell zero is accepting, then $T$ accepts immediately; if it is
rejecting, then $T$ also rejects immediately.
    This yields a simulation of a 1XCA by a (deterministic) TM which is only
polynomially slower, thus implying $\oXCAP \subseteq \PTIME$.
    \qed
\end{proof}

\section{Conclusion}
\label{sec_conclusion}

This paper summarized the results of \textcite{MA} and also presented related
and previously unpublished results from \textcite{BA}.
The main result was the characterization $\XCAP = \ttpNP$
(Theorem~\ref{thm_XCAP_ttpNP}) in Section~\ref{sec_XCAP}, which also gave an
alternative characterization based on NTMs (Theorem~\ref{thm_simulNP}).
In Section~\ref{sec_implications}, XCAs with multiple accept and reject states
were shown to be equivalent to the original model (Theorem~\ref{thm_marxca}).
Also in Section~\ref{sec_implications}, two other variants based on varying
acceptance conditions were considered: the existential (EXCA), in which a
single, though arbitrary cell may direct the automaton's response; and the
one-cell-decision XCA (1XCA), in which a fixed cell does so.
In the first case, it was shown that the polynomial-time class $\EXCAP$ equals
$\XCAP$ (Theorem~\ref{thm_EXCAP_in_ttpNP}); in the latter, it was shown that the
polynomial-time class $\oXCAP$ of 1XCAs equals $\PTIME$
(Theorem~\ref{thm_1XCAP_P}).

This paper has covered some XCA variants with diverse acceptance conditions.
A topic for future work might be considering further variations in this sense
(e.g., XCAs whose acceptance condition is based on \emph{majority} instead
of \emph{unanimity}, which appears to lead to a model whose polynomial-time
class equals $\PP$).
Another avenue of research lies in restricting the capabilities of XCAs and
analyzing the effects thereof (e.g., restricting 1XCAs or SXCAs to a polynomial
number of cells).
A final open question is determining what polynomial speedups, if any, 1XCAs
provide with respect to 1CAs.

\subsection*{Acknowledgements}
  I thank Thomas Worsch for his mentoring, encouragement, and support during the
writing of this paper.
  I would also like to thank Dennis Hofheinz for pointing out a crucial mistake
in a preliminary version of this paper as well as the anonymous referees for
their valuable remarks and suggestions.

\begin{refcontext}[sorting=nyt]
\printbibliography
\end{refcontext}

\end{document}